%% file: main.tex
\documentclass[preprint,12pt]{elsarticle}

\usepackage{lineno}
\usepackage{hyperref}
\usepackage[T1]{fontenc}
\usepackage{xspace}
\usepackage{amsmath}
\usepackage{amssymb}
\usepackage{amsthm}
\usepackage{xcolor}
\newtheorem{theorem}{Theorem}
\newtheorem{lemma}[theorem]{Lemma}
\newtheorem{corollary}[theorem]{Corollary}
\newdefinition{rmk}{Remark}

\title{Broadcast via Mobile Agents in a Dynamic Network: Interplay of Graph Properties \& Agents}

\author[1]{William K. {Moses Jr.}}\ead{wkmjr3@gmail.com}

\author[2]{Amanda Redlich}\ead{amanda_redlich@uml.edu}

\author[3]{Frederick Stock}\ead{frederickstock0@gmail.com}
\affiliation[1]{organization={Department of Computer Science, Durham University},
city={Durham},
country={UK}}
\affiliation[2]{organization={Department of Mathematics \& Statistics, University of Massachusetts Lowell},
city={Lowell, MA},
country={USA}}
\affiliation[3]{organization={Department of Computer Science, University of Massachusetts Lowell},
city={Lowell, MA},
country={USA}}

\newcommand{\MSG}{\mathcal{M}}

\newcommand{\smart}{source\xspace}
\newcommand{\dumb}{ignorant\xspace}
\newcommand{\broadcast}{\textsc{Broadcast}}
\newtheorem{definition}{Definition}

\begin{document}

\begin{keyword}
Mobile agents \sep mobile robots \sep broadcast \sep dynamic graph\sep dynamic network 
\end{keyword}

\begin{abstract}
We revisit the problem of \textsc{Broadcast}, introduced by Das, Giachoudis, Luccio, and Markou [OPODIS, 2020], where $k+1$ agents are initially placed on an $n$ node dynamic graph, where $1$ agent has a message that must be broadcast to the remaining $k$ ignorant agents. The original paper studied the relationship between the number of agents needed to solve the problem and the edge density of the graph. The paper presented strong evidence that edge density of a graph, or the number of redundant edges within the graph, may be the correct graph property to accurately differentiate whether $k= o(n)$ ignorant agents (low edge density) or $k = \Omega(n)$ ignorant agents (high edge density) are needed to solve the problem.

In this paper, we show that surprisingly, edge density may not in fact be the correct differentiating property. The original paper presents graphs with edge density $1.1\overline{6}$ that require $\Omega(n)$  agents, however, we construct graphs with edge density $> 1.1\overline{6}$ and develop an algorithm to solve the problem on those graphs using only $o(n)$ agents. We subsequently show that the relationship between edge density and number of agents is fairly weak by first constructing graphs with edge density tending to $1$ from above that require $\Omega(n/f(n))$ agents to solve, for any function $f(n) \to \infty$ as $n \to \infty$. We then construct an infinite family of graphs with edge density $< \rho$ requiring at least $k$ ignorant agents to solve \textsc{Broadcast}, for any $k>1$ and $\rho>1$.

Finally, we investigate different versions of connectivity as possible properties determining the number of agents.   We show that it is possible for a graph to have low (constant) edge connectivity but require a high (linear) number of agents to solve \textsc{Broadcast}. More generally, for any arbitrary $\lambda$, we construct a family of graphs on $n$ nodes with edge connectivity $(n-1)/\lambda$ requiring at least $k=n - 2 \lambda$ ignorant agents to solve \textsc{Broadcast}.   We then show that for any graph with vertex connectivity $\geq 3$ and minimum degree $\delta$, $k \leq \delta -2$ ignorant agents are unable to solve \textsc{Broadcast}. Finally, we show that for any graph containing a bond with $m$ edges, \textsc{Broadcast} is unsolvable when the number of ignorant agents $k \leq m-2$.
\end{abstract}

\maketitle

\section{Introduction}
\label{sec:intro}
\input{intro}

\section{Generalized Theta Graph Algorithm}
\label{sec:algo}
\input{beachballAlg}

\section{Low Density Graphs Requiring Many Agents}
\label{sec:low-density-many-agents}
\input{low-density-many-agents}

\section{Arbitrary Density Graphs Requiring Arbitrary Agents}
\label{sec:arb-density-arb-agents}
\input{arb-density-arb-agents}

\section{Examining Other Graph Properties}
\label{sec:other-graph-properties}
\input{other-graph-properties}

\section{Conclusions}
\label{sec:conc}
\input{conclusions}

\bibliographystyle{elsarticle-num}
\bibliography{giantref}


\newpage
\appendix
\section{Notation}
\label{app:notation}
\input{notation}

\end{document}

%% file: intro.tex
In this paper, we study the specific problem of \broadcast{} in the larger paradigm of mobile agents on a dynamic graph. Mobile agents on a graph is a useful paradigm to study a plethora of real world problems. It can be used to model situations in which actual robots are trying to move through constrained spaces (e.g., a building with rooms and corridors connecting those rooms) and associated problems such as exploring every room in such a constrained space, and subsequently having all robots regroup (gather). It can also be used to model situations on networks where each agent is a virtual agent that moves from computer to computer in order to spread some information or perform some task on a subset of computers. The extension of this paradigm to dynamic graphs allows the scope of problems being studied to increase in a meaningful way as dynamicity in connections exists whether one is modeling the real world (e.g., earthquakes/fires changing the connectivity between rooms in a building) or a virtual one (e.g., two computers in a network become connected/disconnected because a cable is added/cut between them).

Within this paradigm, many problems have been studied, modeling an array of real world issues. One traditional set of lenses to view those problems through is the robot-configuration-oriented lens vs. graph-finding-oriented lens. Robot-configuration-oriented problems require the robots to achieve some specific configuration, such as gathering~\cite{CFPS12,CP02,DKLHPW11,P07}, scattering~\cite{Barriere2009,ElorB11,Poudel18,Shibata:2016}, pattern formation~\cite{SY99}, dispersion~\cite{Augustine:2018}, and
convergence~\cite{CP04}. Graph-finding-oriented problems require the robots to find something or some property of the graph, such as  exploration~\cite{Bampas:2009,Cohen:2008,Das13,Dereniowski:2015,Fraigniaud:2005,MencPU17}, treasure hunting~\cite{MP15}, triangle counting~\cite{CDM24}, dominating set~\cite{CMS23}, and maximal independent set~\cite{PBCM24}. In this framework, we study the problem of \broadcast{} as introduced in Das, Giachoudis, Luccio, and Markou~\cite{DGLM20}. In it, there is initially one agent with a message that is to be transmitted to a number of other agents located on the nodes of the graph. Under the traditional set of lenses, one might view this problem as a robot-configuration-oriented problem.

However, there is another set of lenses through which one may view the above problems. This set of lenses provides additional insight with respect to the larger context of distributed computation via message passing on traditional networks. This provides an additional motivation and understanding for why \broadcast{} is worth studying. Many problems such as gathering, dispersion, and exploration appear to be mobile agent-specific problems without immediate counter-parts in traditional distributed computing. Other problems, such as maximal independent set of a graph and triangle counting are fundamental graph problems typically studied under standard distributed computing message passing models such as LOCAL and CONGEST. 

Through this alternative set of lens, we see that the problem of \broadcast{}, as studied in this paper, acts as an interesting bridge between agent-specific problems and traditional distributed computing. In traditional distributed computing, \broadcast{} requires all nodes to participate in the broadcast of a message originating at one node. However, here, the number of agents is independent of the number of nodes and furthermore the agents can move around on the nodes. By studying such ``bridge'' problems, we may develop a deeper understanding of the links between the mobile  agents on a graph model and standard distributed computing.

In the paper that introduced \broadcast{} in this context, they focused on the question of the minimum number of agents that allow the problem to be solvable. A number of graph classes of varying density were presented with lower and upper bounds. As a result, they mention that this number ``apparently depends on the density of the underlying graph $G$ or, the number of redundant edges in $G$''. However, in this work, we show that with a finer-grained look at edge density, this surprisingly may not always be the case. We then extend our study to other types of graph properties to see if they might provide a better understanding of the minimum number of agents needed.

\subsection{Preliminaries \& Definitions}
\paragraph*{Graph Definitions}
Initially, we are given a graph $G(V,E)$, with node set $V$ and edge set $E$, which serves as the underlying graph. In each round $i \geq 1$, the adversary chooses some (possibly empty) subset of edges $E'$ from $G$ to remove for that round such that the graph $G_i(V,E \setminus E')$ is still connected. The set of all such graphs $G_i$, $i \geq 1$, is said to be the dynamic network $\mathcal{G}$.

The \textit{edge density} $\rho$ of a graph with $m$ edges and $n$ nodes is $m/n$ as is common~\cite{C00}.\footnote{Note that in~\cite{C00}, the definition of the density of the total graph is slightly different as it suited the particular problem looked at in that paper. In particular, the density of any subgraph is defined as it is here while the density of the overall graph is taken as the maximum density over all subgraphs, in line with the problem of densest subgraph that is studied in that paper.}

A \textit{generalized theta graph}, denoted by $\theta_{d_1, d_2, \ldots, d_{\ell}}$, consists of two fixed nodes joined by $\ell$ paths, where each path $i$ contains $d_i \geq 1$ nodes. See Figure~\ref{fig:theta-graph-ex} for an example illustration.

\begin{figure}[h!t]
    \centering
    \includegraphics[scale=.5]{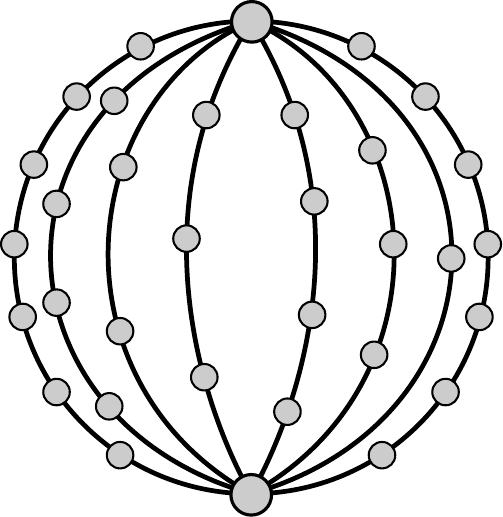}
    \caption{An example theta graph $\theta_{7,4,2,3,4,3,1,7}$}
    \label{fig:theta-graph-ex}
\end{figure}

A \textit{cut} in a graph is a set of edges whose deletion renders the graph disconnected.  The \textit{edge connectivity} of a graph is the size of the smallest cut of the graph. The \textit{vertex connectivity} of a graph is the smallest number of nodes whose deletion renders the graph disconnected. A \textit{bond} of a graph is a set of edges such that if all the edges in the set are removed the graph becomes disconnected, but if all but any one edge is removed the graph is still connected.  In other words, a bond is a minimal cut set of edges.
\paragraph*{Agent Description}
We consider that initially there is one agent 
with the message $\MSG$ and $k$ agents without the message. An agent with knowledge of $\MSG$ is termed a \textit{\smart} agent and an agent without knowledge of $\MSG$ is termed an \textit{\dumb} agent. To keep our analysis consistent with the original context of~\cite{DGLM20}, we maintain the same assumptions on agents used in that paper. We assume that agents have unique IDs and sufficient local memory for computation, and each agent starts at a unique node of $G$. For each $i \geq 1$, each agent has a global view of the graph $G_i$, i.e., each agent sees the complete topology of $G_i$ including any edges the adversary has deleted, and can see which agents occupy which nodes of the graph, including whether the agents are \smart or \dumb agents. As mentioned in~\cite{DGLM20}, these assumptions can be justified by considering that the adversary is very strong and so the agents need fairly strong capabilities in order to be able to solve the problem.

\paragraph*{Communication}
A \smart agent is able to transmit the message $\MSG$ to an \dumb agent when both are co-located at the same node in the same round. In other words, we assume local and not global communication. An \dumb agent becomes a \smart agent when it receives $\MSG$ from another \smart agent.

\paragraph*{Time}
Time proceeds in synchronous rounds. In each round $i \geq 1$, the following occurs in the given order.
\begin{enumerate}
    \item Adversary chooses a (possibly empty) subset of edges $E'$ and removes them from $G$ to create $G'(V, E \setminus E')$.
    \item Agents perform local computation and/or communicate with co-located agents and choose an edge to move through (or possibly decide to stay at the same node).
    \item Agents move through chosen edges.
\end{enumerate}


\paragraph*{Problem Statement}
The problem statement, as stated in~\cite{DGLM20}, is given below.

\begin{definition}[The \broadcast{} Problem]
    Given a constantly connected dynamic network $\mathcal{G}$ based on an underlying graph $G$ consisting of $n \geq 2$ nodes, a \smart agent that has a message $\MSG$ and $k \geq 1$ \dumb agents that are initially located at distinct nodes of the network, the goal is to broadcast this message $\MSG$ to all the agents.
\end{definition}

\subsection{Our Contributions}
In this paper, we focus on a fine-grained analysis of the edge density of a graph with $n$ nodes and its relationship to the number of agents needed to solve \broadcast{} on that graph. The work of Das, Giachoudis, Luccio, and Markou~\cite{DGLM20} used edge density to divide the classes of graphs into sparse and dense graphs and showed that there exist sparse graphs (rings) for which \broadcast{} can be solved using $o(n)$ agents and dense graphs (grids) for which \broadcast{} requires at least $\Omega(n)$ agents. In particular, they show that for rings, with edge density $1$, $k=2$ \dumb agents suffice to solve \broadcast{} and for $2 \times L$ grids, including a $2 \times 3$ grid with edge density $1.1\overline{6}$, $k \geq L = \Omega(n)$ \dumb agents are required to solve \broadcast{}. 

One might naturally assume that for graphs with edge density greater than $1.1\overline{6}$, $k = \Omega(n)$ agents are required to solve \broadcast{}. However, this is not the case. We show that there exists a class of graphs, commonly known as generalized theta graphs, that may have edge density larger than $1.1\overline{6}$ that in fact require only $k= o(n)$ \dumb agents to solve \broadcast{} (Theorem~\ref{the:alg-works} provides a loose bound for any generalized theta graph, Theorem~\ref{the:alg-istight} shows this bound is tight assuming each path has sufficient length).  This surprising result shows that edge density, by itself, may not be the dividing line between graphs that require $k=o(n)$ and $k=\Omega(n)$ \dumb agents. An example of such a construction is as follows. Consider a generalized theta graph with $\ell$ paths between two nodes, such that each path is of length $d$. The number of nodes $n$ is $\ell d + 2$ and the edge density is $1 + (\ell - 2)/(\ell d + 2)$, which can be rewritten as $1 + (n - 2 - 2d)/(nd)$. The number of \dumb agents needed to solve \broadcast{} is $\ell$. When $d = \omega(1)$, we see that $\ell = o(n)$. Let $n=100$ and $d = \log^*n = 4$. Then the given graph requires $o(n)$ \dumb agents and has edge density $1.225 > 1.1\overline{6}$.

A complementary question one might then ask is, whether in the range of edge density  $\rho \in (1, 1.1\overline{6})$, do all graphs with such an edge density require very few agents, i.e., $k \ll n$? We answer this question in the negative by constructing, for any arbitrary $f(n) \to \infty$, a family of graphs with edge density tending to $1$ from above (and therefore in particular $<\rho$) requiring $\Omega(n/f(n))$ agents (Theorem~\ref{the:small-edge-density-many-agents}).

One might then wonder if there is any relationship between the exact edge density of a graph and the exact number of \dumb agents needed to solve for that edge density. From~\cite{DGLM20}, we see that for trees, that have edge density $<1$, $k=1$ suffices. For rings, that have edge density exactly $1$, $k=2$ suffices. For a complete graph, that has edge density $(n-1)/2$, $k \geq n-2$ \dumb agents are needed. The question then becomes: are there specific values $\rho$ and $k$ such that any graph with edge density $< \rho$ requires $< k$ \dumb agents for some constant $k$? We answer this in the negative.  For any $k>0$ and any $\rho >1$, we construct an infinite family of graphs with edge density $<\rho$ requiring at least $k$ \dumb agents to solve \broadcast{} (Theorem~\ref{the:lollipops-density}). In other words, there is no threshold $\rho$ and number of \dumb agents $k$ for which the statement ``any graph whose edge density is below $\rho$ requires $< k$ \dumb agents'' is correct.

A natural next line of inquiry is to look towards other graph properties and ask about their relationship with the number of \dumb agents needed to solve \broadcast{}. A natural candidate is  connectivity, as connectivity is the defining constraint on the adversary.  We show that it is possible for a graph to have low (constant) edge connectivity but require high (linear) number of agents to achieve \broadcast{}. (Contrast this, for example, with the ring, which has edge connectivity $2$ and also requires $2$ agents.)  In fact, for arbitrary $\lambda$, we construct a family of graphs on $n$ vertices with edge connectivity $(n-1)/\lambda$ but requiring at least $n-2\lambda$  \dumb agents for \broadcast{} (Theorem~\ref{the:edge-connectivity}). 

We then look at vertex connectivity and develop an adversary strategy applicable to any graph with vertex-connectivity at least 3. This can be easily applied to many graphs to establish loose lower bounds on the required agents. In particular, for all $\beta \geq 3$, all graphs with vertex connectivity $\beta$ and minimum degree $\delta$ require at least $k \geq \delta - 1$ \dumb agents to solve \broadcast{} (Theorem~\ref{the:vertex-connectivity}).

Finally, we look at the graph property of bonds and show that for any graph which has a bond with $m$ edges, \broadcast{} is unsolvable for $k \leq m-2$ \dumb agents (Corollary~\ref{cor:bonds}).

\subsection{Technical Overview \& Challenges}
While the problem of \broadcast{} seems relatively simple, solving it requires involved algorithms with many small but crucial details. Note, that in the process of solving \broadcast{} the number of \dumb and \smart agents is constantly changing as \dumb agents become \smart agents.
Therefore any winning strategy must adjust for the changing ratio between these two classes of agents.
Hence the algorithm we present in Theorem~\ref{the:alg-works} uses 5 distinct phases, each of which applies when the agents are arranged meeting some condition. Each phase's prerequisite condition is dependent only on the existence of a few \smart agents, and the rest of the agents may be of either class, mitigating complications as \dumb agents are turned into \smart agents. Each phase may either progress to the next phase or repeat an earlier phase depending on the adversary's strategy. However, every time we change phases, measurable progress is made towards solving \broadcast{}.

Although~\cite{DGLM20} observes a potential link between edge density (i.e., ``redundant'' edges) and number of agents needed, we prove that is not always the case.  This is surprising, as more redundant edges means the adversary has more potential strategies available.  In fact one of the main technical challenges was to design graphs and \broadcast{} strategies that are robust against the adversary's many options.  The analysis of generalized theta graphs through the algorithm and lower bound in Section 2 acts as a useful starting point.  However, adjusting internal path lengths and number of paths to simultaneously achieve desired edge densities and agent requirements proved challenging; the generalized theta graph construction is inherently limited to edge densities close to one.  By carefully choosing functions for path length and multiplicity, we are nevertheless able to generate  families of graphs with low edge densities and a large number of required agents (Theorem \ref{the:small-edge-density-many-agents}).

To overcome the limitations of generalized theta graphs and develop constructions for arbitrary densities, we use subgraphs with higher or lower edge density as starting points.  Then we use structural graph techniques to combine these subgraphs, and scale them to create an arbitrary overall average density.  Furthermore, we do this in such a way that the subgraphs are locally well understood for \broadcast{}, and their local \broadcast{} strategies translates to a global strategy for the whole graph.  In particular, we start with cliques and paths for their uniquely high and low densities, and their well-understood properties for \broadcast{}.  We then combine carefully-selected instances of paths and cliques by identifying the two graphs at a single vertex; this preserves edge density and \broadcast{} strategies and allows for a careful analysis (Theorem \ref{the:lollipops-density}).

Finally, the edge and vertex connectivity results required more innovation.  As connectivity is a global (not local) property, it was not possible to merge subgraphs and use their average values.  However, it is possible to combine multiple identical cliques in such a way as to not change their edge-connectivity.  Then the strategies for adversary and \broadcast{} on a larger graph are exactly generated by the strategies for each subclique (Theorem \ref{the:edge-connectivity}).  Another approach is to use other graph theoretic properties such as bonds and degrees, and the high connectivity of complete graphs, to construct algorithms and counterexamples.  In particular, the adversary is able to define a strategy based on a cut set of vertices or edges, and construct a spanning tree of the remaining connected graph that prevents \broadcast{} (Theorems \ref{the:vertex-connectivity}).  

\subsection{Related Work}
Das, Giachoudis, Luccio, and Markou~\cite{DGLM20} introduce the problem of \broadcast{} via mobile agents on dynamic graphs. They study the problem for various classes of graphs and establish upper and lower bounds on the number of agents needed to solve \broadcast{}. In particular, they show that for trees, the problem is always solvable for any $k>0$. For rings, $k\geq 2$ \dumb agents are required except for small rings ($n < 5$). For cactus graphs with $c$ cycles, at least $c+1$ \dumb agents are needed. For grids, $\Omega(n)$ \dumb agents are needed. For the special case of $2 \times L$ grids, they present a lower bound of $\geq L-1$ \dumb agents and an almost matching algorithm requiring $k = L$ \dumb agents. For complete graphs, they show a matching lower bound and algorithm for $n-2$ \dumb agents. For hypercubes, they show a lower bound that $k$ should be almost half the number of nodes in order to solve \broadcast{}. For the special case of a 3-dimensional hypercube, they show a matching lower bound and algorithm that requires $k=n/2$ \dumb agents.

The above work is the only paper looking at the given problem in the given paradigm. If we look at other related work, it is either for a different paradigm or for different problems. When it comes to different paradigms, the problem of \broadcast{} and the related problem of information dissemination have been studied in the traditional distributed message passing setting for a variety of dynamic networks by Awerbuch and Even~\cite{AE84}, Casteigts, Flocchini, Mans, and Santoro~\cite{CFMS15}, Clementi, Pasquale, Monti, and Silvestri~\cite{CPMS09}, Kuhn, Lynch, and Oshman~\cite{KLO10}, and  O'Dell and Wattenhofer~\cite{OW05}.

When it comes to the same paradigm of mobile agents on a dynamic graph, other problems have been solved in this setting. There is usually a divide in the work depending on whether agents know information about the dynamicity of the graph in advance or not. The setting considered in this paper where agents do not know such information apriori is also called the live setting. Di Luna~\cite{DiLuna19} is a good survey on other problems studied under the live setting.

When looking further afield to other paradigms and other problems, there have been many other problems that have been studied in the dynamic graph setting. Examples include the problems of finding matchings~\cite{MertziosMNZZ20},
separators~\cite{FLUSCHNIK2020197,ZSCHOCHE202072}, vertex
covers~\cite{AKRIDA2020108,HammKMS22}, containments of
epidemics~\cite{ENRIGHT202160}, Eulerian tours~\cite{BumpusM23,MarinoS21}, graph
coloring~\cite{MarinoS22,MertziosMZ21}, network flow~\cite{AKRIDA201946},
treewidth~\cite{Fluschnik2020}, and cops and robbers~\cite{Erlebach-Spooner/20,MorawietzRW20}. One may refer to the survey by Michail~\cite{Michail16} for more information.

%% file: beachballAlg.tex
Consider any generalized theta graph $G = \theta_{d_1,d_2, \dots, d_l}$ with $l \geq 1$ paths, each of length $\geq 1$. In this section, we develop an algorithm that solves \broadcast{} with $k \geq l$ \dumb agents (and $1$ \smart agent).
The algorithm consists of an initial pre-processing step (Lemma~\ref{lem:preprocess}) followed by a ``main algorithm'' of 5 phases (Theorem~\ref{the:alg-works}).

The preprocessing step takes any initial configuration of agents and rearranges them such that \dumb agents are in some sense equally distributed among the paths of $G$.
This ensures that the main algorithm always has a somewhat consistent initial configuration of agents.
In the main algorithm, every phase $i \in [2,5]$, has some assumed precondition before the phase can be entered.
Each phase $i \in [1,5]$ ends if the precondition for the next phase is satisfied, resulting in the algorithm moving to the next phase. Phases $i \in [2,5]$ can also end if an \dumb agent becomes a \smart agent, resulting in the algorithm repeating an earlier phase $j \in [1,i-1]$.
At some point, there will be no more \dumb agents remaining in $G$.
At this point, the algorithm terminates and \broadcast{} has been solved.

The key insight of this strategy is as follows. If there is a \smart agent present at each pole of the $G$, then these agents can pick any path $P$ of $G$ and make any \dumb agent on $P$ a \smart agent. The adversary can only remove at most one edge of $P$ without disconnecting $G$. Each \smart agent can then walk from its pole, along $P$ towards the other pole. As the adversary cannot remove two edges of $P$, the \smart agents will eventually traverse all but possibly one edge of $P$, visiting every vertex of $P$, making every \dumb agent on $P$ a \smart.
Phase 5 specifically makes use of this fact to turn an \dumb agent into a \smart. The rest of the phases all slowly work toward creating this scenario, creating situations where the adversary must either allow an \dumb agent to meet a \smart, or allow the algorithm to progress to the next phase. 

Below is our first Lemma, which establishes the preprocessing phase of our algorithm.

\begin{lemma}\label{lem:preprocess}
    Regardless of the initial configuration of at least $l$ agents on a generalized theta graph $G = \theta_{d_1, d_2, \dots, d_l}$, a configuration can always be reached such that
    one path has two \smart agents and every other path has a distinct \dumb agent.
\end{lemma}
\begin{proof}

First, we arrange the agents such that there is one path $P$ of $G$ that has the initial \smart agent on one of its internal vertices and every path (including $P$) has one distinct \dumb agent.

Suppose this condition is not met, there are only two reasons it might not be. The \smart agent is on a pole, and/or some path has at least two \dumb agents on its internal vertices. 

In the first case, the \smart agent can be moved to at least one path of $G$, as the adversary cannot disconnect the pole from every path of $G$. In the second case, assume we have two \dumb agents on one path $P$. Each agent is closer to one endpoint of $P$ than the other agent, move each agent towards that endpoint. The adversary cannot stop one of them from leaving $P$. For any path with more than one \dumb agent, we repeat this process until there is only one \dumb agent. These processes give us the desired configuration.

At this point we do the following:
Again, let $P$ be the path of $G$ that contains the \smart agent, call it $s$ and let $d$ be the \dumb agent that is identified with $P$. Let $N$ and $S$ (``north'' and ``south'') represent the two poles of $G$. Without loss of generality, assume $d$ is closer to $S$ than $s$ is ($s$ is ``sandwiched'' by $N$ and $d$). We move every agent in $G$ (\dumb and \smart) towards $N$, forcing one of two outcomes. Either agent $s$ and an additional ignorant agent reach $N$ (where it becomes a source), or $s$ never reaches $N$ but $d$ meets it and becomes a source.

Note that at every timestep at least one path will have no edge removed (otherwise the graph would be disconnected). Therefore at least one path's ignorant agent may make progress towards $N$. Thus the sum of distances between ignorant agents and $N$ is monotonically decreasing, and eventually one ignorant agent must reach $N$.

If the adversary chooses to remove an edge between $s$ and $d$, then $s$ and $d$ will never meet on $P$. But $s$ will progress towards $N$, due to our choice that $d$ is further from $N$ than $s$. As we just noted, another \dumb agent will also reach $N$, creating the first outcome, two agents meet at $N$. If the adversary instead decides to remove an edge from $P$ that blocks $s$ from reaching $N$, then $s$ will eventually get ``stuck'' somewhere on $P$. Hence, $d$ will eventually catch up to $s$ as the adversary cannot remove another edge from $P$ without disconnecting $G$, giving us the latter outcome. 
In either case we have two \smart agents that can be identified with the same path.
\end{proof}

We now present Theorem~\ref{the:alg-works}, proving solvablity of \textsc{Broadcast} on generalized theta graphs by presenting a winning strategy. 
Clearly, this will require some detailed descriptions of agents moving on these graphs, therefore, we start by creating some notation.
Any path $P$ of a generalized theta graph $G$ augmented with the fixed nodes $N, S$, can be written in the form, $N, v_1, v_2, \dots, v_{m}, S$.
Where its two endpoints $N$, $S$ are the ``poles'' of $G$, and are shared by every path of $G$.
We call the vertices $\{v_1, v_2, \dots, v_m\}$ the \emph{internal} vertices of $P$.
We call $N$ and $S$ the \emph{poles} of $G$.
An agent is \emph{identified} with $P$ if it is currently on some vertex of $\{N, v_1, v_2, \dots, v_{m}, S\}$. 
For any agent on an internal vertex of a path this is well-defined but it is ambiguous for any agent on $N$ or $S$.
We say an agent may be identified with only one path, so if an agent is on $N$ or $S$, then there is one unique path it is identified with.
If an agent is on a pole we may decide to change its identified path, this will be explicitly stated, otherwise we assume its identified path remains the same.
We may refer to an agent and its identified path by saying the agent is on a path, or the path the agent is on.
If no agent is identified with a path, then that path is \emph{empty}.

\begin{figure}
    \centering
    \includegraphics[width=1\linewidth]{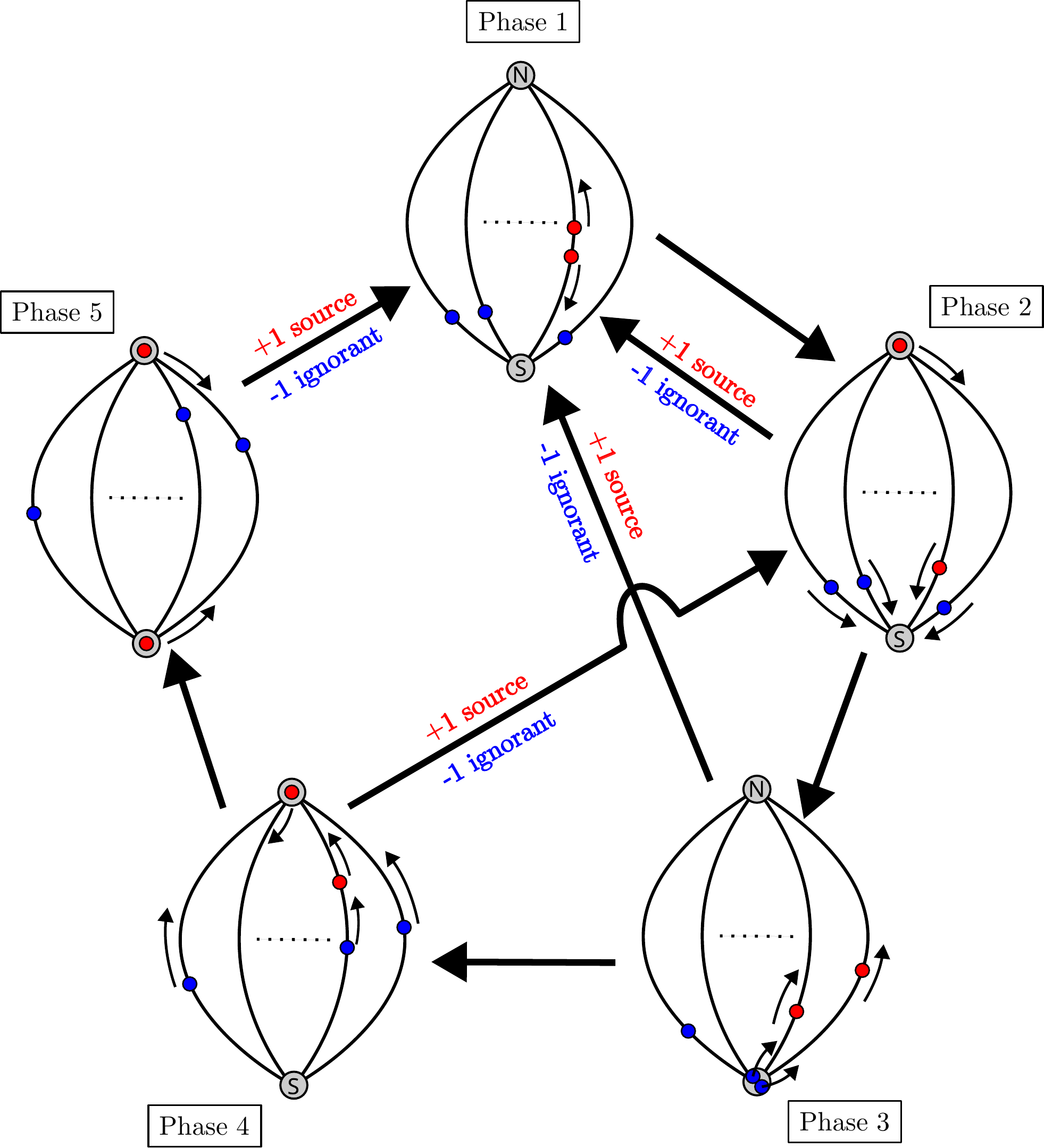}
    \caption{Operation of the algorithm presented in Theorem~\ref{the:alg-works}, blue nodes are \dumb agents, and red are \smart agents. Arrows possible phase transitions, with labels indicating if the quantity of either class of agent changes.}
    \label{fig:beachball-alg}
\end{figure}

\begin{theorem}
\label{the:alg-works}
    Given a generalized theta graph $G = \theta_{d_1, d_2,...,d_l}$ with $l \geq 1$ paths, each of length $\geq 1$. \textsc{Broadcast} can be solved on $G$ with $k\geq l$ \dumb agents.
\end{theorem}

\begin{proof}

    We present a 5-phase algorithm to solve \broadcast{}. 
    Each phase is initiated when a prerequisite condition is met, each listed below:
    \newcounter{phasecounterfirst}
    \renewcommand{\thephasecounterfirst}
          {({\ifcase\value{phasecounterfirst}\or \textbf{Phase 1}\or \textbf{Phase 2}\or \textbf{Phase 3}\or \textbf{Phase 4}\or \textbf{Phase 5}\else junk\fi\relax})}
    
    \begin{list}{\thephasecounterfirst}{\usecounter{phasecounterfirst}}
        \item A path has two source agents.
        \item A \smart agent on one pole and there is no empty path.
        \item Two \dumb agents on one pole and two paths have an \textbf{internal} \smart agent.
        \item At least one \smart agent $s$ on one pole, a different \smart agent $s'$ on a path, and at least one empty path.
        \item A \smart agent on each pole.
    \end{list}

    The algorithm will start in Phase 1, and then progress through each phase, in a non-linear manner, as demonstrated by Figure~\ref{fig:beachball-alg}.    
    We maintain that every phase will either (1) progress to a phase of higher index (i.e. phase 1 to phase 2) or (2) an \dumb agent will transform into a \smart agent while returning to a lower index phase (i.e. phase 4 to phase 2). 
    Therefore, this algorithm will eventually terminate with every \dumb agent being turned into a \smart agent.
    The phase transitions are demonstrated in Figure~\ref{fig:beachball-alg}, with an example of how the agents are moved in each.
    This process ends when all agents verify that there are no \dumb agents left on $G$.

    To simplify our algorithm, we assume that there are exactly $l$ \dumb agents (i.e, $k=l$). If there are more than $l$ \dumb agents the following algorithm could be easily adapted. At the start of the algorithm, take a random subset of $l$ \dumb agents and pretend the rest do not exist. Then run the algorithm with these $l$ \dumb agents. After the entirety of this subset has been turned to \smart agents, ``forget'' those agents and draw a new subset of $l$ \dumb agents from those remaining. If there are less than $l$ remaining \dumb agents, select \smart agents to make up the difference. Then start the algorithm at the first applicable phase.

    We also assume that the agents are dispersed across $G$ as described by Lemma~\ref{lem:preprocess}. This is a trivial assumption, if the agents are not arranged in this manner, just apply the lemma. 
    Therefore, we can assume the existence of two \smart agents which we call $s$ and $s'$.

\newcounter{phasecounter}
\renewcommand{\thephasecounter}
      {({\ifcase\value{phasecounter}\or \textbf{Phase 1}\or \textbf{Phase 2}\or \textbf{Phase 3}\or \textbf{Phase 4}\or \textbf{Phase 5}\else junk\fi\relax})}

\begin{list}{\thephasecounter}{\usecounter{phasecounter}}

    \item\label{phase:1} \textbf{A path has two source agents.}
    
    Recall we have at least two \smart agents $s$ and $s'$.
    The objective of this phase is to move to \ref{phase:2}. This requires $s$ and $s'$ be arranged so one is on a pole of $G$ and the other is on an internal vertex of some vertex of a path of $G$. If both \smart agents are on $N$, move one \smart agent off of $N$ onto any path. Otherwise, both agents are still on $P$, we move the two source agents in opposite directions, one towards $S$ and one towards $N$. Since the adversary can only remove one edge from $P$ they can't stop one of them from reaching a pole. When we have one agent on a pole, the algorithm can progress to \ref{phase:2}.

    \item\label{phase:2} \textbf{A \smart agent on one pole and there is no empty path.}
    
    Now we can assume, with no loss of generality, there is one \smart agent $s$ on $N$, and there is some path $P$ that has a different \smart agent $s'$. Move every agent (\smart and \dumb) toward $S$. Fix a path from $N$ to $S$ containing an \dumb agent, move $s$ along this path. One of two things happens, either a \smart agent meets an \dumb agent or an \dumb agent reaches $S$. We now use an analysis almost symmetric to that in Lemma~\ref{lem:preprocess} to show at least two \dumb agents will reach $S$, or a \smart agent will meet with a \dumb agent. 
    
    If the adversary wishes to prevent $s$ from reaching the \dumb agent on the chosen path, it must remove an edge between the two agents, guaranteeing that the \dumb agent will reach $S$. If the adversary does not, the \dumb agent will be turned into a \smart agent by $s$, in the latter case the algorithm returns to \ref{phase:1}, with one more \smart agent, and one less \dumb agent. Therefore, we assume the former case, the \dumb agent reaches $S$.
    By the connectivity constraint, the adversary must leave at least one path unaltered at each timestep, and as each path has at least one agent, an additional \dumb agent will reach $S$. There are now two \dumb agents on $S$, and two paths have \smart agents, the algorithm progresses to \ref{phase:3}. 

    \item\label{phase:3} \textbf{Two \dumb agents on one pole and two paths have an \textbf{internal} \smart agent.}
    
    We have two \dumb agents $d$ and $d'$ on a pole, assume $S$. We start by moving every agent towards $N$. Every agent but $d$ and $d'$ moves along the path it is identified with. Move $d$ on the path that contains $s$ and $d'$ on the path with $s'$. We claim either a \smart agent meets an \dumb agent on a path, or a \smart agent reaches $N$. As we assumed there were $\ell$ \dumb agents and one \smart agent and there are two paths with two agents, there is a path of $G$ that is empty. Therefore we can no longer assume at every time step one agent will progress towards $N$. However, on the two paths that have both an \dumb and a \smart agent, the adversary must make a choice. They allow the \smart agent to move towards $N$, or the \dumb agent moves closer to the \smart agent. If they do allow one of the \smart agents to move towards $N$, it will eventually reach $N$, and the algorithm progresses to \ref{phase:4}. On the other hand, if $s$ (resp., $s'$) is prevented from reaching $N$, then $d$ (resp., $d'$)  cannot be prevented from reaching $s$ ($s'$), as doing so would require removing two edges from one path. In this case either $d$ or $d'$ will eventually become a \smart, and we can return to phase~\ref{phase:1}.

    \item\label{phase:4} \textbf{At least one \smart agent $s$ on one pole, a different \smart agent $s'$ on a path, and at least one empty path.}
    
    Let $s$ be at pole $N$. Further let $s'$ be the \smart agent on a path.
    This phase will either progress to \ref{phase:2} or \ref{phase:5}.
    
    First, we ensure that the number of \smart agents on $N$ is greater than or equal to the number of empty paths.
    To this end, suppose there are two empty paths and only one \smart agent on $N$. 
    By the pigeonhole principle, there is at least one path with two agents, neither of which is $s$.
    We can use the process of \ref{phase:1} to move one of those agents to a pole and on to another path. 
    When this process moves one agent onto a pole, there are four cases. 
    (1), a \smart agent moves onto $S$, (2) a \smart agent moves onto $N$, (3) a \dumb agent moves onto $S$, and (4) a \dumb agent moves onto $N$.
    In the first case, we have a \smart agent on both $S$ and $N$, and we move to \ref{phase:5}. 
    In the other three cases, either a \dumb agent moves to $S$, and we can identify it with an empty path, or, an agent other than $s$ reaches $N$, putting an additional \smart agent on $N$ (as $s$ would turn an \dumb agent into a \smart agent). 
    Therefore, by repeated applications of this argument, we can assume the number of \smart agents on $N$ is equal to the number of empty paths.

    Now that the number of \smart agents on $N$ outnumber the empty paths, we move to the second phase of our algorithm. We can assume $s'$ is on a path with an \dumb agent, such that $s'$ is \textit{sandwiched} between the \dumb agent and $N$.
    If $s'$ is on a path with a \dumb agent but it is not sandwiched, then this \dumb agent is between $s'$ and $s$ (as $s$ is assumed to be on $N$). Both $s$ and $s'$ can move towards each other to eventually reach the \dumb agent and make it a \smart agent, we then return to \ref{phase:1}. 
    If there is no \smart agent that shares a path with an \dumb agent, then by the assumed number of \dumb agents, there must be a path with two \dumb agents.
    One of these agents can be moved to a pole using the technique from~\ref{phase:1}, and one will either reach
    $N$, becoming \smart, and we can move to \ref{phase:1}. Otherwise, one will reach $S$, giving us our sandwiched condition.
    Therefore we may assume $s'$ exists and is sandwiched by an \dumb agent and $N$.

    For the rest of the description of this phase, we describe the movements of one \smart agent $s$ on $N$, on one empty path. Otherwise, if there are multiple \smart agents on $N$ and empty paths we move the extra \smart agents analogously to how we move $s$ in the following. 
    Now, move every agent besides $s$ towards $N$. Move $s$ towards $S$ along an empty path, if there is no empty path, pick one arbitrarily, prioritizing paths that do not contain a \smart agent. Either $s$ reaches $S$ and another \smart agent, $s'$, reaches $N$. We show the adversary cannot block both from happening. 
    
    $s'$ has a \dumb agent $d'$ on its path, so to prevent $s'$ from reaching $N$ would allow $d'$ to meet $s'$. Therefore we assume $s'$ reaches $N$, as otherwise we could progress to \ref{phase:1}. By preventing $s$ from reaching $S$, the adversary would be unable to prevent some \dumb agent from reaching $N$, as preventing this would require removing an edge from every path, disconnecting $G$. Now, by the assumption that $s$ will also reach $N$ this \dumb agent will become a \smart, placing two \smart agents at $N$, and we can progress to~\ref{phase:2}. 

    Therefore, either we make additional \smart agents and repeat the process at the corresponding step, or we successfully move one \smart agent to both poles of $G$, moving to \ref{phase:5}. 

    \item\label{phase:5} \textbf{A \smart agent on each pole.}
    Pick a path containing a \dumb agent. Move the \smart agents along this path towards this \dumb agent.
    One will eventually reach the \dumb agent, as the adversary could only prevent this by removing two edges simultaneously from the path, which would disconnect $G$. When the \dumb agent becomes a \smart agent, we progress to \ref{phase:1}.
\end{list}

\end{proof}

We can show that when the path lengths of $G$ are all at least 3, the previous theorem presents a tight bound.

\begin{theorem}\label{the:alg-istight}
    Given a generalized theta graph $G = \theta_{d_1, d_2,...,d_l}$ with $l \geq 1$ paths, each of length $\geq 3$. \textsc{Broadcast} can be solved on $G$ if and only if there are $k\geq l$ \dumb agents.  
\end{theorem}
\begin{proof}
    Clearly, the fact that \broadcast{} can be solved with at least $k = l$ \dumb agents is established by Theorem~\ref{the:alg-works}. Therefore we establish an adversary strategy when there are $k < l$ many \dumb agents.
    Assume there are $l - 1$ \dumb agents, as this is the hardest case for the adversary.
    There are exactly $l$ many agents, so the adversary can initially place exactly one agent on each path, specifically on an internal vertex of each path. 
    Then the adversary prevents every \dumb agent from reaching either pole of $G$, by removing an edge from each path between the \dumb agent and its closest pole. 
    The only agent that can therefore leave its path is the \smart agent. 
    However if the \smart agent does reach a pole and then enter another path with an \dumb agent on it, the adversary can prevent them from meeting.
    The adversary can then remove an edge between the \smart and \dumb agent. 
    This would then allow the \dumb agent to leave the path and reach a pole. 
    However, this then produces an indefinitely repeating process, one agent reaches a pole and moves to a path with another agent. 
    These agents will not be able to reach each other, and one may reach a pole.
    But now there is one agent on a pole which then moves to a path with a different agent\dots 
    During this loop the adversary can prevent all other agents from reaching a pole and interfering with this process.
\end{proof}

%% file: low-density-many-agents.tex
Here we give a construction for a family of generalized theta graphs with low density but high agent requirements, where both quantities are in terms of $n$, the number of vertices.  As the number of vertices grows, the edge density decreases but the required number of agents grows (albeit at a slower rate than the number of vertices).  Most specifically, it is a family of graphs such that the edge density is $1+\epsilon(n)$ and the number of agents required is $\delta(n)$ for functions $\epsilon(n) \to 0$ and $\delta(n) \to \infty$.  We accomplish this by choosing the dimensions of a generalized theta graph strategically.

\begin{figure}
    \centering
    \includegraphics[width=0.5\linewidth]{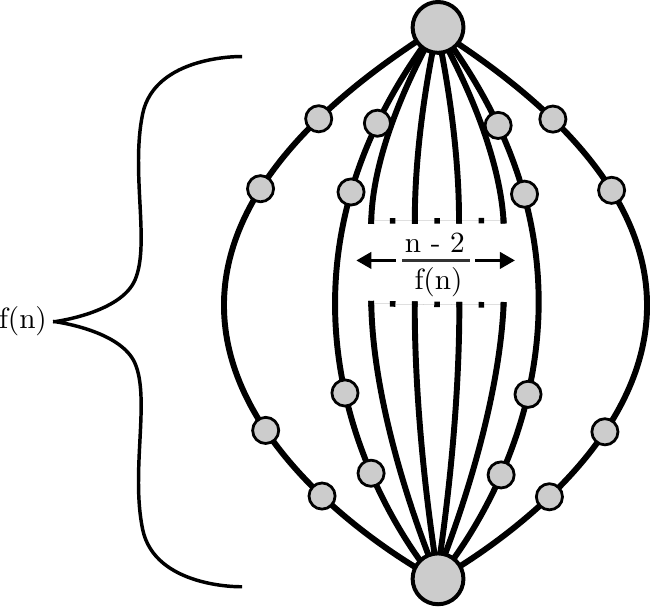}
    \caption{An example of the construction from Theorem~\ref{the:small-edge-density-many-agents}. A theta graph with $(n-2)/f(n)$ paths of $f(n)$ vertices each.}
    \label{fig:thrm-4}
\end{figure}

\begin{theorem}\label{the:small-edge-density-many-agents}
    For any function $f(n) \to \infty$, there exists a family of graphs with $n$ vertices and $n(1+1/f(n))-2(1+1/f(n))$ edges requiring at least $k= (n-2)/f(n)$ \dumb agents to solve \broadcast{}. 
\end{theorem}
\begin{proof}
Consider the generalized theta graph on $n$ vertices with $(n-2)/f(n)$ paths on  $f(n)$ vertices each.  This graph has $(f(n)+1)((n-2)/f(n))=n-2+(n-2)/f(n)=n(1+1/f(n))-2(1+1/f(n))$ edges, and therefore edge density $1+1/f(n)-(2/n)(1+1/f(n))$.  As we already know, this graph requires at least $(n-2)/f(n)$ \dumb agents to solve \broadcast{}.
\end{proof}

This generic theorem can be implemented with different functions of $n$ as convenient, depending on the desired ratio of agents to vertices.  For example, the following two corollaries give (a) a family of graphs where the number of agents grows polynomially and (b) a family of graphs where the number of agents grows nearly linearly.

\begin{corollary}
    For any $\epsilon$ there is a family of graphs with edge density $<1+\Theta (n^{-\epsilon})$ requiring at least $k = \Theta(n^{1-\epsilon})$ \dumb agents to solve \broadcast{}.
\end{corollary}
\begin{proof}
  Let $f(n)=n^{\epsilon}$ in the above proof.  
\end{proof}

\begin{corollary}
    There is a family of graphs with edge density $1+O(1/\log^*(n))$ requiring at least $k = \Theta (n/\log^* (n))$ \dumb agents to solve \broadcast{}.
\end{corollary}
\begin{proof}
    Let $f(n)=\log^*(n)$ in the above proof.
\end{proof}

%% file: arb-density-arb-agents.tex
We now turn to a new type of construction, designed to create graphs with exact arbitrary values of agents required and upper bounds on graph density.  This is in contrast to the previous construction; in theory, since the limit of required agents $\to \infty$ and the edge density $\to 1$, for $n$ sufficiently large the number of agents will be larger than and the edge density less than any desired pair of constants.  However this is an \emph{existence} statement without giving an efficient method of finding the necessary size of $n$.

On the other hand, here we give a direct construction that allows for precise values for $k$, the number of ignorant agents, and arbitrarily low edge density.  This we do by introducing a new family of graphs and making some key observations about strategies on certain types of subgraphs.

\begin{theorem}
\label{the:lollipops-density}
For any natural number $k>1$ and real number $\rho > 1$, there is a graph with edge density $\leq \rho$ that requires at least $k$ \dumb agents to solve \broadcast{}.
\end{theorem}
\begin{proof}

\begin{figure}
    \centering
    \includegraphics[width=0.7\linewidth]{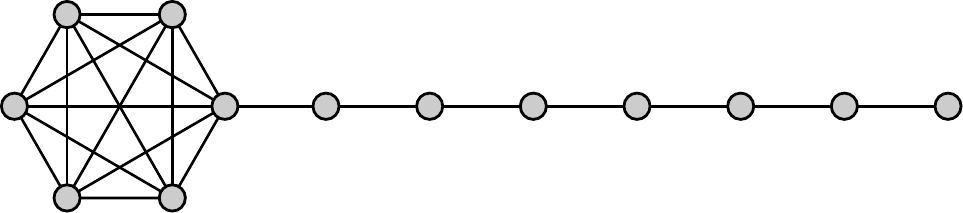}
    \caption{An example of the construction from Theorem~\ref{the:lollipops-density}. A lollipop graph with a clique $K_{k+3} = K_6$ and $s = 7$ edges in the path, and edge density $\leq 1.69231$. At least $k = 3$ ignorant agents are required to solve \broadcast{} on it. If we wish to reduce the edge density even further while maintaining the value of $k$, we may increase the value of $s$.}
    \label{fig:thrm-7}
\end{figure}

    Construct a ``lollipop graph'', that is, a graph constructed by identifying one node in a clique and using it as the the endpoint of a path.  For ease of notation, we will use $K_{k+3}$ for the clique on $k+3$ vertices and $P_s$ for the path with $s$ edges in this construction. For an illustrative example, refer to Figure~\ref{fig:thrm-7}, where $k = 3$ and $s = 7 $.  Then the constructed graph has $n = s+k+3$ vertices and $s+\binom{k+3}{2}$ edges.  The edge density of this graph is thus 
    \begin{align}
        \frac{s +\binom{k+3}{2}}{s+k+3}
        =\frac{2s+k^2+5k+6}{2s+2k+6}
    \end{align}
    which we want to be $\leq \rho$.  Observe that as $s \to \infty$, this fraction tends to $1$.  Therefore for $s$ sufficiently large it must be $< \rho$ for any $\rho >1$.  Straightforward computation finds the specific lower bound on $s$:
    \begin{align}
    \frac{2s+k^2+5k+6}{2s+2k+6} \leq \rho\\
   2s+k^2+5k+6 \leq \rho(2s+2k+6)\\
   2(1-\rho)s\leq 2\rho k +6\rho -k^2-5k-6\\
   s\geq \frac{2\rho k+ 6\rho-k^2-5k-6}{2(1-\rho)}=\frac{k^2+5k+6-2\rho k - 6\rho}{2(\rho-1)}
    \end{align}

    (Note the inequality reverses direction when dividing by $2(1-\rho)$, a negative number). If $1<\rho<2$, letting $s=\lceil\frac{k^2+5k+6}{\rho-1}\rceil$ is enough.  For $\rho\geq 2$, $s=k^2+5k+6$ suffices.

    If there are fewer than $k$ \dumb agents, the adversary prevents broadcast from occurring by arranging all \dumb agents and the \smart agent in the $K_{k+3}$ subgraph as described in the proof of Theorem~16 in~\cite{DGLM20}.

   If there are $k$ or more \dumb agents, the agents solve broadcast by initially moving into the $K_{k+3}$ subgraph (they may do this because path edges may not be removed by the adversary, so any agents on the path may move to the clique).  Once all agents are in the $K_{k+3}$ subgraph, they follow the strategy described in the proof of Theorem~16 in~\cite{DGLM20}.
\end{proof}

%% file: other-graph-properties.tex
After exploring the relationship between edge density and the number of agents required, we now move on to other interesting graph properties. As the adversary's only restriction is leaving the graph connected in each round, a natural question is ``does a graph's connectivity determine the number of agents required?'' 

Here we study how connectivity influences the adversary's strategy.  We give some graph conditions in terms of connectivity and degree which generate a winning strategy for the adversary, i.e., a strategy that allows the adversary to prevent \broadcast{} from being solved, and explicitly give that strategy.  For some of these graphs, we also give a complementary winning strategy for \broadcast{} when the number of agents is above an explicit threshold.

We give a construction of a family of graphs with lower edge connectivity than agents required, where the connectivity and required agents are functions of each other.  We first state it with maximum generality and then give a few example corollaries with convenient values for connectivity and agents required.

\begin{theorem}\label{the:edge-connectivity}
For any natural numbers $\lambda>1$ and $n> 2\lambda$, there exists a graph on $n$ vertices with edge connectivity $\lfloor \frac{n-1}{\lambda}\rfloor$ and with edge density $\frac{(n-1-(\lambda \lfloor \frac{n-1}{\lambda}\rfloor))\binom{\lceil \frac{n-1}{\lambda}\rceil+1}{2}+(\lambda-(n-(\lambda \lfloor \frac{n-1}{\lambda}\rfloor)))\binom{\lfloor \frac{n-1}{\lambda}\rfloor+1}{2}}{n}\simeq \frac{n}{2\lambda}$ requiring at least $k = n-2\lambda$ \dumb agents to solve \broadcast{}. \end{theorem}

\begin{proof}  Divide $n-1$ vertices into $\lambda$ parts as in the Tur\'an graph; that is, $(n-1-(\lambda \lfloor \frac{n-1}{\lambda}\rfloor))=r$ (the remainder when $n-1$ is divided by $\lambda$) groups with $\lceil \frac{n-1}{\lambda}\rceil=n/\lambda +\epsilon_1$ vertices for some $-1<\epsilon_1<1$ and $\lambda-(n-(\lambda \lfloor \frac{n-1}{\lambda}\rfloor))$ groups with $\lfloor \frac{n-1}{\lambda}\rfloor=\frac{n}{\lambda}+\epsilon_2$ vertices  for some $-1<\epsilon_2<1$.  However, construct the complement of the Tur\'an graph, i.e., each part becomes its own clique.  Then add one more vertex which is adjacent to all $n-1$ vertices. An example with $n=18$ and $\lambda=3$ is given in Figure \ref{fig:thm8}.

\begin{figure}
    \centering
    \includegraphics[width=.7\linewidth]{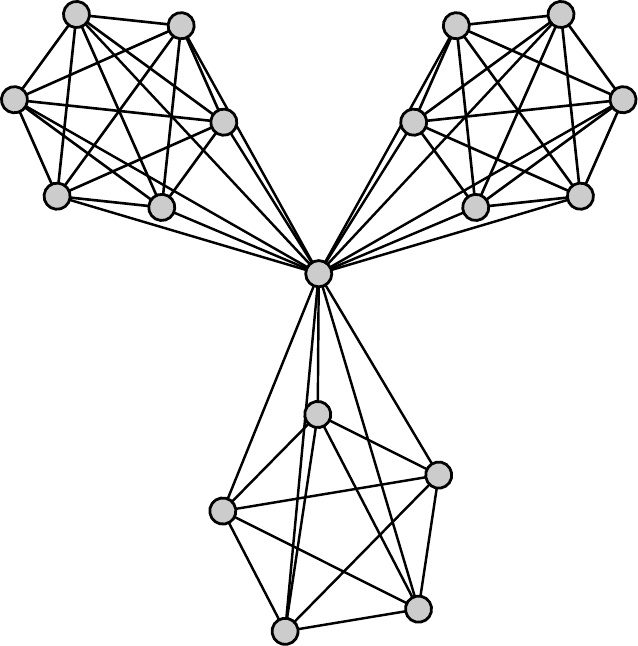}
    \caption{Graph constructed when $n=18$ and $\lambda=3$.  Observe that two of the subcliques have $\lceil 17/3\rceil=6$ vertices and one has $\lfloor 17/3 \rfloor=5$ vertices.}
    \label{fig:thm8}
\end{figure}

The edge density of this graph is  \begin{align}
  \frac{(r\binom{\lceil \frac{n-1}{\lambda}\rceil+1}{2}+(\lambda - r)\binom{\lfloor \frac{n-1}{\lambda}\rfloor+1}{2}}{n}\\
  =\frac{r\left(\left(\lceil \frac{n-1}{\lambda}\rceil+1\right)\left(\lceil \frac{n-1}{\lambda}\rceil\right)\right)+\left(\lambda-r\right)\left(\lfloor \frac{n-1}{\lambda}\rfloor+1\right)\left(\lfloor \frac{n-1}{\lambda}\rfloor\right)}{2n}\\
  =\frac{r\left(\left(\frac{n}{\lambda}+\epsilon_1+1\right)\left(\frac{n}{\lambda}+\epsilon_1\right)\right)+\left(\lambda-r\right)\left(\frac{n}{\lambda}+\epsilon_2+1\right)\left( \frac{n}{\lambda} +\epsilon_2\right)}{2n}\\
 =\frac{n}{2\lambda}+\epsilon
\end{align}
where $\epsilon$ is some error term with absolute value $\leq 3/2$. 

 Furthermore, the adversary can follow a winning strategy similar to that in the proof of Theorem \ref{the:lollipops-density} by placing $j-2$ \dumb agents in each clique of size $j$, and the agents can successfully broadcast once at least one subclique contains $j-1$ agents. 
 A pigeonhole argument shows that the number of \dumb agents needed is at least $n-2\lambda$.
\end{proof}

\begin{corollary}
For arbitrary $\delta$ there is a family of graphs on $n$ vertices, $ n \to \infty$, with edge density \emph{and} edge connectivity $\Theta(n^{\delta})$  requiring at least $k = n-\Theta(n^{1-\delta})$ \dumb agents to solve \broadcast{}.
\end{corollary}
\begin{proof}
    Let $\lambda=n^{1-\delta}/2$ in the above theorem.
\end{proof}

\begin{corollary}
For arbitrary $0<c<1$ there is a family of graphs on $n$ vertices, $n \to \infty$, with constant edge density \emph{and} edge connectivity, requiring at least $cn-1$ \dumb agents.
\end{corollary}
\begin{proof}
Let $\lambda=\lfloor\frac{n-cn+1}{2}\rfloor$ in the above theorem.
\end{proof}


 We now turn our attention to vertex connectivity. Inspired by Lemma~17 of \cite{DGLM20}, we present a simple tool to establish a lower bound on $k$ for any graph that is at least 3-vertex-connected.  

\begin{theorem}\label{the:vertex-connectivity}
    Given a 3-vertex-connected graph $G = (V,E)$ with minimum degree $\delta$, \textsc{Broadcast} is unsolvable with $k \leq \delta - 2$ \dumb agents.
\end{theorem}

\begin{proof}
    The adversary can initially pick a random vertex $s \in V$ and place the \smart agent there. 
    The degree of $s$ is at least $\delta$, however, there are only $\delta - 2$ many \dumb agents, so there are at least two neighbors of $s$ that do not have an agent. The adversary then picks one of these. Without loss of generality, call it $u$. Then $u$ has at least $\delta$ neighbors, and there are at most $\delta - 2$ many agents (\dumb or \smart) other than the agent at $s$ in $G$. Therefore, there is a vertex other than $s$ adjacent to $u$ that does not have an agent. Call this vertex $r$. Since $G$ is 3-connected, the adversary can compute a spanning tree on $G \setminus \{u,s\}$ rooted at $r$. The adversary deletes all edges apart from the edges $\{s, u\}$ and $\{u, r\}$ and the edges in this spanning tree. Now, there are two unoccupied vertices between $s$ and any other occupied vertex of $G$. No matter the position of the \smart agent, this is always possible, hence the adversary can repeat this process indefinitely and an \dumb agent will never meet the \smart. 
\end{proof}


Finally, we study the power of a specific type of bond in a graph for a generalized situation where there might be multiple \smart agents initially. In order to proceed, we remind  the reader of the definition of a \textit{bond} of a graph, which is a set of edges such that if all the edges in the set are removed the graph becomes disconnected, but if all but one edge is removed the graph is still connected.  In other words, a bond is a minimal cut set of edges.  Here we consider \emph{only} bonds that are also matchings, i.e., the endpoints of the bond's edges are all distinct.  We call these ``matching bonds''.

\begin{theorem}\label{the:bonds}
If a graph $G$ has a matching bond with $m$ edges, then the adversary has a winning strategy when there are initially $\leq k_1$ \dumb agents and $\leq k_2$ \smart agents for any $k_1+k_2 \leq m-1$. 

\end{theorem}

\begin{proof}
    For ease of notation, let the edges in the matching bond be $u_1v_1, u_2v_2, \ldots u_mv_m$ where all the $u_i$ are on one side, $A$, of the bond and the $v_i$ are all on the other side, $B$.  Recall that the adversary chooses the initial position of all agents.  In this instance, the adversary places the \dumb agents on distinct vertices in $A$ and the \smart agents on distinct vertices in $B$.  (Because the bond has $m$ edges, there are at least $m$ vertices in $A$ and $m$ in $B$; thus there are enough vertices to do this).

    At every round, the adversary removes any edge $u_iv_i$ that contains an agent on $u_i$ or $v_i$.  Observe that for any agent to cross from $A$ to $B$ (respectively, $B$ to $A$), it must first reach some $u_i$ ($v_i$) vertex and then move to some $v_i$ ($u_i$).  Therefore this adversary strategy prevents any agent initially placed in $A$ from reaching $B$ or vice-versa. Furthermore, since at most $m-1$ edges of the bond are removed at any step, the graph remains connected. Since \dumb agents may only visit vertices in $A$ and \smart agents may only visit vertices in $B$, this proves that the adversary wins.
\end{proof}

We give a corollary for the situation in the original \broadcast{} problem where there is initially only $1$ \smart agent.
\begin{corollary}\label{cor:bonds}
    If a graph $G$ has a matching bond with $m$ edges, then the adversary has a winning strategy for $\leq k$ \dumb agents when $k \leq m-2$.
\end{corollary}

%% file: conclusions.tex
In this paper, we pursued a more fine-grained study of edge density and its relationship to the number of agents needed to solve \broadcast{}. We showed that, surprisingly, edge density does not clearly act as a differentiator for graphs as to whether $o(n)$ or $\Omega(n)$ agents are needed to solve \broadcast{}. After exploring this property in detail, we turned our attention to various properties of connectivity. We were able to draw some interesting relationships between the number of agents needed and edge connectivity, vertex connectivity, and bond size. 

However, it is still not clear what graph property, if any, acts as a clear distinguisher for the number of agents required to solve the problem. We hope that by highlighting that this problem is not as clear cut as previously thought, we may spur further research into this and into the relationship of various other graph properties to the solvability of \broadcast{}.

%% file: notation.tex
Below is a centralized reference for the notation found throughout the paper:
\begin{itemize}
    \item $\theta_{d_1, d_2, \ldots, d_{\ell}}$ - generalized theta graph that consists of two fixed nodes joined by $\ell$ paths, where each path $i$ contains $d_i \geq 1$ internal nodes. Totally, the graph has $\sum_i^\ell d_i + 2$ nodes and edge density $\sum_i^\ell (d_i + 1)/(\sum_i^\ell d_i + 2)$.
    \item \emph{pole} - A term used for theta graphs to refer to the two nodes that are the endpoints of each path
    \item $N$ and $S$ - typical notation used to represent the poles of a theta graph (North and South)
    \item $\rho$ - edge density of a graph.
    \item $\epsilon$ - used sometimes when talking about edge density, in particular $\rho = 1 + \epsilon$.
    \item $n$ - number of nodes in a graph.
    \item $m$ - number of edges in a graph.
    \item $k$ - number of \dumb agents initially present in a graph.
    \item $\delta$ - minimum degree of any node of a graph
    \item $\MSG$ - message present with \smart agents.
    \item \broadcast{} - the problem where $k$ \dumb agents and $1$ \smart agent are initially placed at distinct locations on the graph, and the final goal is for the $k$ \dumb agents to receive a message and become \smart  agents. Often when talking about the problem, we only mention the number of \dumb agents, $k$, initially as it is assumed that we always start with $1$ \smart agent.
\end{itemize}

%% file: giantref.bib
@inproceedings{DGLM20,
  title={Broadcasting with mobile agents in dynamic networks},
  author={Das, Shantanu and Giachoudis, Nikos and Luccio, Flaminia L. and Markou, Euripides},
  booktitle={24th International Conference on Principles of Distributed Systems (OPODIS 2020)},
  year={2020},
  organization={Schloss Dagstuhl-Leibniz-Zentrum f{\"u}r Informatik}
}

@inproceedings{OW05,
  title={Information dissemination in highly dynamic graphs},
  author={O'Dell, Regina and Wattenhofer, Rogert},
  booktitle={Joint Workshop on Foundations of Mobile Computing (DIALM-POMC)},
  pages={104--110},
  year={2005},
  organization={ACM}
}

@inproceedings{KLO10,
  title={Distributed computation in dynamic networks},
  author={Kuhn, Fabian and Lynch, Nancy and Oshman, Rotem},
  booktitle={ACM symposium on Theory of computing (STOC)},
  pages={513--522},
  year={2010},
  organization={ACM}
}

@inproceedings{CDM24,
  author       = {Prabhat Kumar Chand and
                  Apurba Das and
                  Anisur Rahaman Molla},
  editor       = {Mehdi Dastani and
                  Jaime Sim{\~{a}}o Sichman and
                  Natasha Alechina and
                  Virginia Dignum},
  title        = {Agent-Based Triangle Counting and Its Applications in Anonymous Graphs},
  booktitle    = {Proceedings of the 23rd International Conference on Autonomous Agents
                  and Multiagent Systems, {AAMAS} 2024, Auckland, New Zealand, May 6-10,
                  2024},
  pages        = {2186--2188},
  publisher    = {International Foundation for Autonomous Agents and Multiagent Systems
                  / {ACM}},
  year         = {2024},
  url          = {https://dl.acm.org/doi/10.5555/3635637.3663102},
  doi          = {10.5555/3635637.3663102},
  bibsource    = {dblp computer science bibliography, https://dblp.org}
}

@inproceedings{PBCM24,
  author       = {Debasish Pattanayak and
                  Subhash Bhagat and
                  Sruti Gan Chaudhuri and
                  Anisur Rahaman Molla},
  title        = {Maximal Independent Set via Mobile Agents},
  booktitle    = {Proceedings of the 25th International Conference on Distributed Computing
                  and Networking, {ICDCN} 2024, Chennai, India, January 4-7, 2024},
  pages        = {74--83},
  publisher    = {{ACM}},
  year         = {2024},
  url          = {https://doi.org/10.1145/3631461.3631543},
  doi          = {10.1145/3631461.3631543},
  bibsource    = {dblp computer science bibliography, https://dblp.org}
}

@inproceedings{CMS23,
  author       = {Prabhat Kumar Chand and
                  Anisur Rahaman Molla and
                  Sumathi Sivasubramaniam},
  editor       = {Konstantinos Georgiou and
                  Evangelos Kranakis},
  title        = {Run for Cover: Dominating Set via Mobile Agents},
  booktitle    = {Algorithmics of Wireless Networks - 19th International Symposium,
                  {ALGOWIN} 2023, Amsterdam, The Netherlands, September 7-8, 2023, Revised
                  Selected Papers},
  series       = {Lecture Notes in Computer Science},
  volume       = {14061},
  pages        = {133--150},
  publisher    = {Springer},
  year         = {2023},
  url          = {https://doi.org/10.1007/978-3-031-48882-5\_10},
  doi          = {10.1007/978-3-031-48882-5\_10},
  bibsource    = {dblp computer science bibliography, https://dblp.org}
}

@inproceedings{AE84,
  title={Efficient and reliable broadcast is achievable in an eventually connected network},
  author={Awerbuch, Baruch and Even, Shimon},
  booktitle={Proceedings of the third annual ACM symposium on Principles of distributed computing},
  pages={278--281},
  year={1984}
}

@article{CFMS15,
  title={Shortest, fastest, and foremost broadcast in dynamic networks},
  author={Casteigts, Arnaud and Flocchini, Paola and Mans, Bernard and Santoro, Nicola},
  journal={International Journal of Foundations of Computer Science},
  volume={26},
  number={4},
  pages={499--522},
  year={2015},
  publisher={World Scientific}
}

@inproceedings{CPMS09,
  title={Information spreading in stationary markovian evolving graphs},
  author={Clementi, Andrea EF and Pasquale, Francesco and Monti, Angelo and Silvestri, Riccardo},
  booktitle={2009 IEEE International Symposium on Parallel \& Distributed Processing},
  pages={1--12},
  year={2009},
  organization={IEEE}
}

@inproceedings{C00,
  title={Greedy approximation algorithms for finding dense components in a graph},
  author={Charikar, Moses},
  booktitle={International workshop on approximation algorithms for combinatorial optimization},
  pages={84--95},
  year={2000},
  organization={Springer}
}

@article{Michail16,
  author       = {Othon Michail},
  title        = {An Introduction to Temporal Graphs: An Algorithmic Perspective},
  journal      = {Internet Math.},
  volume       = {12},
  number       = {4},
  pages        = {239--280},
  year         = {2016},
  nourl          = {https://doi.org/10.1080/15427951.2016.1177801},
  doi          = {10.1080/15427951.2016.1177801},
  bibsource    = {dblp computer science bibliography, https://dblp.org}
}

@inproceedings{Erlebach-Spooner/20,
  author       = {Thomas Erlebach and
                  Jakob T. Spooner},
  title        = {A Game of Cops and Robbers on Graphs with Periodic Edge-Connectivity},
  booktitle    = {46th International Conference on Current Trends in Theory and Practice of Informatics (SOFSEM 2020)},
  series       = {Lecture Notes in Computer Science},
  volume       = {12011},
  pages        = {64--75},
  publisher    = {Springer},
  year         = {2020},
  nourl          = {https://doi.org/10.1007/978-3-030-38919-2\_6},
  doi          = {10.1007/978-3-030-38919-2\_6},
  bibsource    = {dblp computer science bibliography, https://dblp.org}
}

@InProceedings{MorawietzRW20,
  author =	{Nils Morawietz and Carolin Rehs and Mathias Weller},
  title =	{{A Timecop’s Work Is Harder Than You Think}},
  booktitle =	{45th International Symposium on Mathematical Foundations of Computer Science (MFCS 2020)},
  pages =	{71:1--71:14},
  series =	{Leibniz International Proceedings in Informatics (LIPIcs)},
  year =	{2020},
  volume =	{170},
  noeditor =	{Javier Esparza and Daniel Kr{\'a}ľ},
  publisher =	{Schloss Dagstuhl--Leibniz-Zentrum f{\"u}r Informatik},
  address =	{Dagstuhl, Germany},
  doi =		{10.4230/LIPIcs.MFCS.2020.71},
}

@incollection{Fluschnik2020,
	author = {Fluschnik, Till and Molter, Hendrik and Niedermeier, Rolf and Renken, Malte and Zschoche, Philipp},
	booktitle = {Treewidth, Kernels, and Algorithms: Essays Dedicated to Hans L. Bodlaender on the Occasion of His 60th Birthday},
	doi = {10.1007/978-3-030-42071-0_6},
	noeditor = {Fomin, Fedor V. and Kratsch, Stefan and van Leeuwen, Erik Jan},
	pages = {49--77},
	publisher = {Springer International Publishing},
	title = {As Time Goes By: Reflections on Treewidth for Temporal Graphs},
	year = {2020},
  }

@inproceedings{MarinoS21,
	author = {Marino, Andrea and Silva, Ana},
	booktitle = {Combinatorial Algorithms},
	noeditor = {Flocchini, Paola and Moura, Lucia},
	pages = {485--500},
  doi = {10.1007/978-3-030-79987-8_34},
	publisher = {Springer International Publishing},
	title = {K{\"o}nigsberg Sightseeing: Eulerian Walks in Temporal Graphs},
	year = {2021}
}

@article{MarinoS22,
	author = {Andrea Marino and Ana Silva},
	doi = {10.1016/j.jcss.2021.08.004},
	journal = {Journal of Computer and System Sciences},
	pages = {171-185},
	title = {Coloring temporal graphs},
	volume = {123},
	year = {2022},
  }

@article{AKRIDA201946,
	author = {Eleni C. Akrida and Jurek Czyzowicz and Leszek G{\k a}sieniec and {\L}ukasz Kuszner and Paul G. Spirakis},
	doi = {10.1016/j.jcss.2019.02.003},
	journal = {Journal of Computer and System Sciences},
	pages = {46-60},
	title = {Temporal flows in temporal networks},
	volume = {103},
	year = {2019},
  }

@article{HammKMS22,
	author = {Hamm, Thekla and Klobas, Nina and Mertzios, George B. and Spirakis, Paul G.},
	doi = {10.1609/aaai.v36i9.21259},
	journal = {Proceedings of the AAAI Conference on Artificial Intelligence},
	month = {Jun.},
	number = {9},
	pages = {10193-10201},
	title = {The Complexity of Temporal Vertex Cover in Small-Degree Graphs},
	volume = {36},
	year = {2022},
  }

@article{MertziosMZ21,
	author = {George B. Mertzios and Hendrik Molter and Viktor Zamaraev},
	doi = {10.1016/j.jcss.2021.03.005},
	journal = {Journal of Computer and System Sciences},
	pages = {97-115},
	title = {Sliding window temporal graph coloring},
	volume = {120},
	year = {2021},
  }

@article{ZSCHOCHE202072,
	author = {Philipp Zschoche and Till Fluschnik and Hendrik Molter and Rolf Niedermeier},
	doi = {10.1016/j.jcss.2019.07.006},
	journal = {Journal of Computer and System Sciences},
	pages = {72-92},
	title = {The complexity of finding small separators in temporal graphs},
	nourl = {https://www.sciencedirect.com/science/article/pii/S0022000019300546},
	volume = {107},
	year = {2020},
}

@article{FLUSCHNIK2020197,
	author = {Till Fluschnik and Hendrik Molter and Rolf Niedermeier and Malte Renken and Philipp Zschoche},
	doi = {10.1016/j.tcs.2019.03.031},
	journal = {Theoretical Computer Science},
	pages = {197-218},
	title = {Temporal graph classes: A view through temporal separators},
	nourl = {https://www.sciencedirect.com/science/article/pii/S0304397519301975},
	volume = {806},
	year = {2020},
  }

@article{ENRIGHT202160,
	author = {Jessica Enright and Kitty Meeks and George B. Mertzios and Viktor Zamaraev},
	doi = {10.1016/j.jcss.2021.01.007},
	journal = {Journal of Computer and System Sciences},
	pages = {60-77},
	title = {Deleting edges to restrict the size of an epidemic in temporal networks},
	volume = {119},
	year = {2021},
  }

@article{AKRIDA2020108,
	author = {Eleni C. Akrida and George B. Mertzios and Paul G. Spirakis and Viktor Zamaraev},
	doi = {10.1016/j.jcss.2019.08.002},
	journal = {Journal of Computer and System Sciences},
	pages = {108-123},
	title = {Temporal vertex cover with a sliding time window},
	volume = {107},
	year = {2020},
  }

@InProceedings{MertziosMNZZ20,
  author =	{George B. Mertzios and Hendrik Molter and Rolf Niedermeier and Viktor Zamaraev and Philipp Zschoche},
  title =	{{Computing Maximum Matchings in Temporal Graphs}},
  booktitle =	{37th International Symposium on Theoretical Aspects of Computer Science (STACS 2020)},
  pages =	{27:1--27:14},
  series =	{Leibniz International Proceedings in Informatics (LIPIcs)},
  noISBN =	{978-3-95977-140-5},
  noISSN =	{1868-8969},
  year =	{2020},
  volume =	{154},
  noeditor =	{Christophe Paul and Markus Bl{\"a}ser},
  publisher =	{Schloss Dagstuhl--Leibniz-Zentrum f{\"u}r Informatik},
  address =	{Dagstuhl, Germany},
  noURL =		{https://drops.dagstuhl.de/opus/volltexte/2020/11888},
  URN =		{urn:nbn:de:0030-drops-118881},
  doi =		{10.4230/LIPIcs.STACS.2020.27}
}

@article{DiLuna19,
  title={Mobile agents on dynamic graphs},
  author={Di Luna, Giuseppe Antonio},
  journal={Distributed Computing by Mobile Entities: Current Research in Moving and Computing},
  pages={549--584},
  year={2019},
  publisher={Springer},
  doi={10.1007/978-3-030-11072-7_20}
}

@article{BumpusM23,
  title={Edge exploration of temporal graphs},
  author={Bumpus, Benjamin Merlin and Meeks, Kitty},
  journal={Algorithmica},
  volume={85},
  number={3},
  pages={688--716},
  year={2023},
  publisher={Springer},
  doi={10.1007/s00453-022-01018-7}
}

@article{Das13,
  author    = {Shantanu Das},
  title     = {Mobile agents in distributed computing: Network exploration},
  journal   = {Bulletin of the {EATCS}},
  volume    = {109},
  pages     = {54--69},
  year      = {2013},
}

@article{MP15,
  author    = {Avery Miller and
               Andrzej Pelc},
  title     = {Tradeoffs between cost and information for rendezvous and treasure
               hunt},
  journal   = {J. Parallel Distributed Comput.},
  volume    = {83},
  pages     = {159--167},
  year      = {2015},
}

@inproceedings{DKLHPW11,
  author    = {Bastian Degener and
               Barbara Kempkes and
               Tobias Langner and
               Friedhelm {Meyer auf der Heide} and
               Peter Pietrzyk and
               Roger Wattenhofer},
  title     = {A tight runtime bound for synchronous gathering of autonomous robots
               with limited visibility},
  booktitle = {Proc. of the 23rd Annual {ACM} Symposium on Parallelism
               in Algorithms and Architectures (SPAA)},
  pages     = {139--148},
  year      = {2011},
}

@article{SY99,
  author    = {Ichiro Suzuki and
               Masafumi Yamashita},
  title     = {Distributed Anonymous Mobile Robots: Formation of Geometric Patterns},
  journal   = {{SIAM} J. Comput.},
  volume    = {28},
  number    = {4},
  pages     = {1347--1363},
  year      = {1999},
}

@inproceedings{CP04,
  author    = {Reuven Cohen and
               David Peleg},
  title     = {Robot Convergence via Center-of-Gravity Algorithms},
  booktitle = {Proc. of the 11th International
               Colloquium on Structural Information and Communication Complexity, (SIROCCO)},
  pages     = {79--88},
  year      = {2004},
}

@article{P07,
  author    = {Giuseppe Prencipe},
  title     = {Impossibility of gathering by a set of autonomous mobile robots},
  journal   = {Theor. Comput. Sci.},
  volume    = {384},
  number    = {2-3},
  pages     = {222--231},
  year      = {2007},
}

@inproceedings{CP02,
  author    = {Mark Cieliebak and
               Giuseppe Prencipe},
  title     = {Gathering Autonomous Mobile Robots},
  booktitle = {Proc. of the 9th International Colloquium on Structural
               Information and Communication Complexity (SIROCCO)},
  pages     = {57--72},
  year      = {2002},
}

@article{CFPS12,
  author    = {Mark Cieliebak and
               Paola Flocchini and
               Giuseppe Prencipe and
               Nicola Santoro},
  title     = {Distributed Computing by Mobile Robots: Gathering},
  journal   = {{SIAM} J. Comput.},
  volume    = {41},
  number    = {4},
  pages     = {829--879},
  year      = {2012},
}

@inproceedings{Poudel18,
  author    = {Pavan Poudel and
               Gokarna Sharma},
  title     = {Time-Optimal Uniform Scattering in a Grid},
  booktitle = {ICDCN},
 year = {2019},
 pages     = {228--237},
 }

@inproceedings{Barriere2009,
  author={L. Barriere and P. Flocchini and E. Mesa-Barrameda and N. Santoro},
  booktitle={IPDPS},
  title={Uniform scattering of autonomous mobile robots in a grid},
  year={2009},
  pages={1-8},
}

@article{ElorB11,
  author    = {Yotam Elor and
               Alfred M. Bruckstein},
  title     = {Uniform multi-agent deployment on a ring},
  journal   = {Theor. Comput. Sci.},
  volume    = {412},
  number    = {8-10},
  pages     = {783--795},
  year      = {2011}
}

@inproceedings{Shibata:2016,
 author = {Shibata, Masahiro and Mega, Toshiya and Ooshita, Fukuhito and Kakugawa, Hirotsugu and Masuzawa, Toshimitsu},
 title = {Uniform Deployment of Mobile Agents in Asynchronous Rings},
 booktitle = {PODC},
 year = {2016},
 isbn = {978-1-4503-3964-3},
 location = {Chicago, Illinois, USA},
 pages = {415--424},
}

@article{Cohen:2008,
 author = {Cohen, Reuven and Fraigniaud, Pierre and Ilcinkas, David and Korman, Amos and Peleg, David},
 title = {Label-guided Graph Exploration by a Finite Automaton},
 journal = {ACM Trans. Algorithms},
 issue_date = {August 2008},
 volume = {4},
 number = {4},
 month = aug,
 year = {2008},
 issn = {1549-6325},
 pages = {42:1--42:18},
}

@article{Fraigniaud:2005,
 author = {Fraigniaud, Pierre and Ilcinkas, David and Peer, Guy and Pelc, Andrzej and Peleg, David},
 title = {Graph Exploration by a Finite Automaton},
 journal = {Theor. Comput. Sci.},
 issue_date = {22 November 2005},
 volume = {345},
 number = {2-3},
 month = nov,
 year = {2005},
 issn = {0304-3975},
 pages = {331--344},
}

@article{Dereniowski:2015,
 author = {Dereniowski, Dariusz and Disser, Yann and Kosowski, Adrian and Pajak, Dominik and Uzna\'{n}ski, Przemyslaw},
 title = {Fast Collaborative Graph Exploration},
 journal = {Inf. Comput.},
 issue_date = {August 2015},
 volume = {243},
 number = {C},
 month = aug,
 year = {2015},
 issn = {0890-5401},
 pages = {37--49},
}

@article{MencPU17,
  author    = {Artur Menc and
               Dominik Pajak and
               Przemyslaw Uznanski},
  title     = {Time and space optimality of rotor-router graph exploration},
  journal   = {Inf. Process. Lett.},
  volume    = {127},
  pages     = {17--20},
  year      = {2017},
}

@inproceedings{Bampas:2009,
 author = {Bampas, Evangelos and G{a}sieniec, Leszek and Hanusse, Nicolas and Ilcinkas, David and Klasing, Ralf and Kosowski, Adrian},
 title = {Euler Tour Lock-in Problem in the Rotor-router Model: I Choose Pointers and You Choose Port Numbers},
 booktitle = {DISC},
 year = {2009},
 isbn = {3-642-04354-2, 978-3-642-04354-3},
 location = {Elche, Spain},
 pages = {423--435},
}

@article{Augustine:2018,
  author    = {Augustine, John and Moses Jr., William K.},
  title     = {Dispersion of Mobile Robots: {A} Study of Memory-Time Trade-offs},
  journal   = {CoRR},
  volume    = {abs/1707.05629},
  year      = {[v4] 2018 (a preliminary version appeared in ICDCN'18)}
}
